\documentclass[USenglish]{article}\usepackage{amsmath,amsthm,amssymb,fullpage,graphicx,url,xspace,cite,quotes,algorithm2e}
\usepackage[hidelinks]{hyperref}
\newtheorem{theorem}{Theorem}[section]\newtheorem{lemma}[theorem]{Lemma}\newtheorem{corollary}[theorem]{Corollary}

\newcommand{\W}{\ensuremath W\xspace}

\newcommand{\wvps}{weakly visible polygons\xspace}
\newcommand{\wvp}{weakly visible polygon\xspace}
\newtheorem{definition}{Definition}

\newcommand{\comment}[1]{{}}\newcommand{\old}[1]{{}}

\title{An Overview of Minimum Convex Cover and Maximum Hidden Set}
\author{Reilly Browne}

\date{January 2024}

\begin{document}

\maketitle

\begin{abstract}
    We give a review of results on the minimum convex cover and maximum hidden set problems. In addition, we give some new results. First we show that it is NP-hard to determine whether a polygon has the same convex cover number as its hidden set number. We then give some important examples in which these quantities don't always coincide. Finally, We present some consequences of insights from Browne, Kasthurirangan, Mitchell and Polishchuk \cite{Browne2023} on other classes of simple polygons.
\end{abstract}

\section{Preliminaries}

We study two optimization problems within polygon visibility, specifically the \textit{minimum convex cover problem} and the \textit{maximum hidden set problem}.
The \textit{minimum convex cover problem} asks for the smallest cardinality set $B$ of convex polygons whose union is equal to some target polygon $P$. We will refer to the convex polygons as ``convex pieces'' so that the distinction between the covering polygons and the target polygon is clear. The \textit{maximum hidden set problem} is the largest cardinality set $H$ of points within a target polygon $P$ such that no two points $p,q \in H$ can see each other. Here, two points $p,q$ see each other if and only if the line segment $pq$ is contained entirely inside of the target polygon $P$.

Shermer\cite{shermer_1989} introduced the notion of a point visibility graph, an infinite graph for which the vertices correspond to the points of a polygon and edges exist between vertices if the corresponding points see each other. This lets us redefine the minimum convex cover problem and the maximum hidden set problem in terms of graph theory, namely with respect to the clique cover and independent set problems

\begin{definition}\label{definition1}[\cite{shermer_1989}]
Given a polygon $P$, the \textbf{point-visibility graph of P}, $PVG(P) = (V, E)$ where $V = \{p \; | \; p \in P\}$ and $E = \{(x, y) \; | \; \overline{xy} \subset P\}$.
\end{definition}

\begin{definition}\label{definition4}
For a polygon P, the \textbf{hidden set number of P}, hs(P), is the independence number of PVG(P).
\end{definition}

\begin{definition}\label{definition5}
For a polygon P, the \textbf{convex cover number of P}, cc(P), is the minimum number of convex pieces needed to cover P. \text{If $P$ is simple, then} it is also the clique covering number of PVG(P). 
\end{definition}

From these graph theory definitions, it becomes immediately obvious that the following inequality holds:
$$hs(P) \leq cc(P)$$

This inequality holds for polygons with holes as well, where the graph theoretical properties of convex cover do not hold up. It is not always the case in polygons with holes that cliques in the PVG correspond to convex pieces in the polygonal domain, specifically when the convex hull of the points corresponding to the clique vertices contains a hole. This is important for the correctness of Abrahamsen, Meyling, and Nusser's \cite{abrahamsen_et_al:LIPIcs.SoCG.2023.66} heuristic for minimum convex cover in polygons with holes. 

\section{Hardness Results}

For simple polygons, finding the minimum convex cover is known to be APX-hard \cite{EidenbenzConvexCover} i.e. there is no PTAS unless $P=NP$. This is also the case for maximum hidden set \cite{EIDENBENZHidden}. In addition to hardness of approximation, minimum convex cover has also been show to be $\exists \mathbb{R}$-hard by Abrahamsen \cite{Abrahamsen2022CoveringPI}. When the problems are generalized to polygons with holes, we do not know any additional hardness results for minimum convex cover, but maximum hidden set becomes more difficult to approximate. Eidenbenz\cite{EIDENBENZHidden} shows that there unless $P=NP$, then there exists an $\epsilon > 0$ such that there is no polynomial time approximation that achieves a better factor than $O(n^{\epsilon})$.

\subsection{Determining the Gap}
It is also NP-hard to determine the difference $cc(P) - hs(P)$. We define a \textit{homestead polygon} to be a polygon for which $cc(P) = hs(P)$. Therefore, it suffices to show the decision problem of determining whether or not a polygon is a homestead to be NP-hard. 

The key component in the two NP-hardness reductions \cite{shermer_1989} \cite{culbreck} (both from 3-SAT, the variant of the Boolean satisfiability problem where all clauses have 3 literals) for hidden set and convex cover in simple polygons is the construction of a simple polygon which is a homestead if the 3-SAT instance is satisfiable and a nonhomestead if it is not satisfiable. 

\begin{theorem}
Deciding if a simple polygon is a homestead polygon or not is NP-hard.
\end{theorem}
\begin{proof}

It is sufficient to show that for any instance of 3-SAT, the corresponding construction from Shermer's \cite{shermer_1989} NP-hardness reduction for hidden set has exactly the convex cover number $k$ which, if equal to the hidden set number, indicates that the instance is satisfied. This is because we can simply apply any algorithm for deciding if a polygon is a homestead polygon to this construction in the same way that Shermer takes the value $k$ as input for a maximum hidden set decision algorithm. The full details of the construction are provided by Shermer.

For a 3-SAT instance with m clauses and n variables, $k = 2mn + 8n + m + 1$. We show that there is a subgraph of every construction that has clique cover number equal to $K$. Shermer shows a set of convex pieces which cover the construction of size $k$, so we just need to prove that no cover with less pieces can exist using the subgraph argument. First we show the subgraph in Figure \ref{fig:fig21}, with dotted edges signifying that the edge connects to a point in a different component. We use the same labels as in Shermer's proof.

\begin{figure}[ht]
    \centering
\includegraphics[width = 300px]{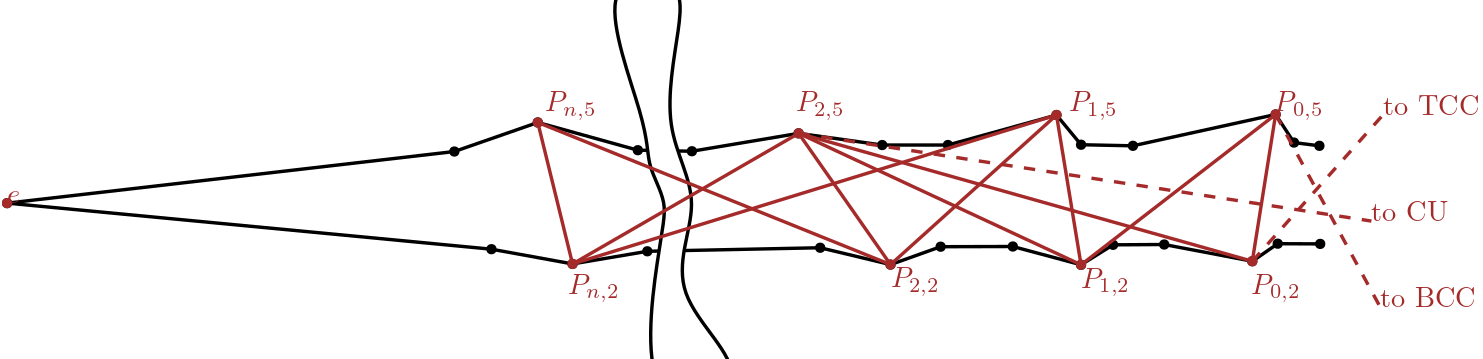}
\includegraphics[width = 100px]{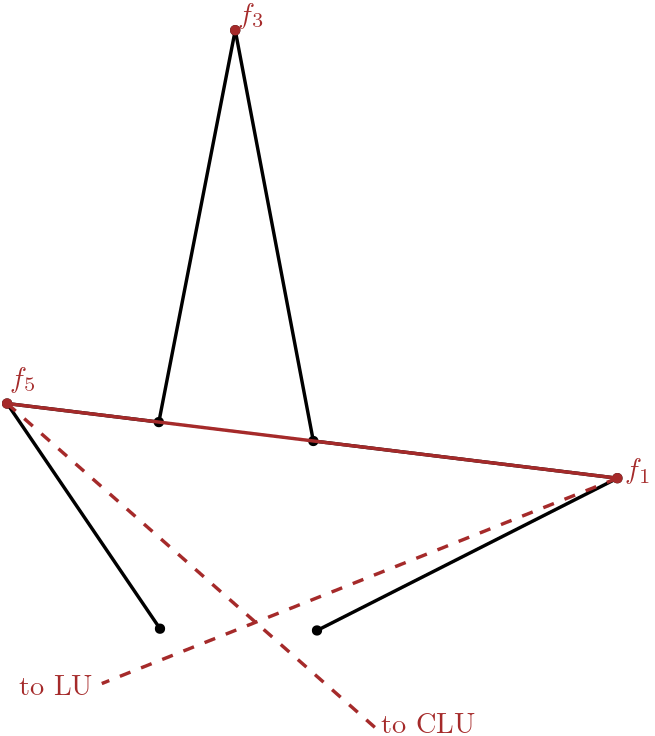}
\includegraphics[height = 100px]{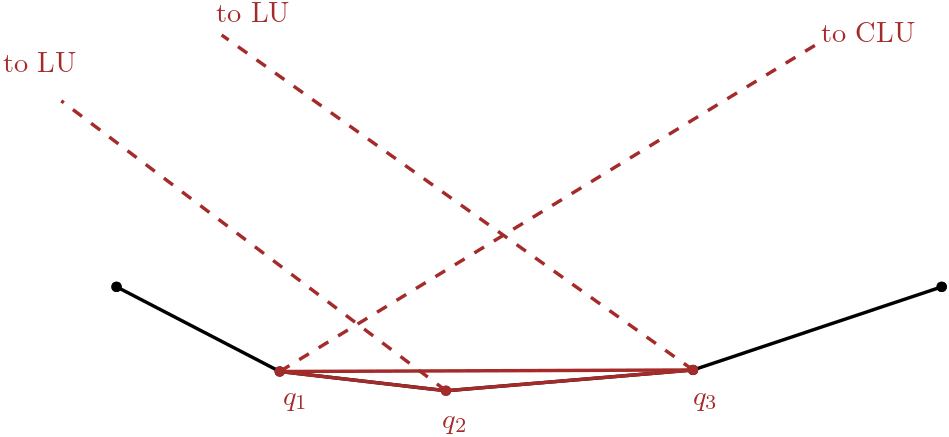}
\includegraphics[height = 100px]{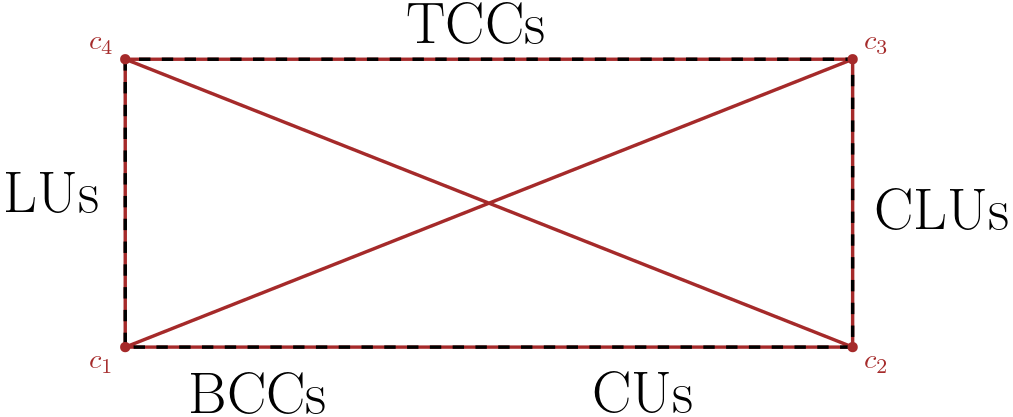}
    \caption{The subgraphs in each of the shermer construction components. Top is a literal unit, left is a consistency checker, right is a clause unit and bottom is the overall structure, indicating we are including the corners.}
    \label{fig:fig21}
\end{figure}

Just as in the previous examples, we start by finding a large clique in the complement graph since the chromatic number of the complement must be at least the size of that clique. A clique in the complement graph is an independent set in the graph, so we will select a set of points identified in Figure \ref{fig:fig20} such that no two share an edge.

For each literal unit, we take all the $e$ points and the $P_{i,2}$ points for $0 \leq i \leq n$, amounting to $2mn+4n$ points. For each consistency checker, we take all the $f$ for an additional $2n$ points. For each clause checker, we take all the $q_0$ points for an additional $m$ points. And from the overall box structure, we take $c_1$. The total size of this clique is $2mn + 6n + m + 1$, meaning we will have to show that an additional $2n$ colors are needed to color the complement graph.

For a start, we can label $c_2,c_3,c_4$ the same color as $c_1$ without any issue. For each $P_{i,5}$ in each literal unit, we can label $P_{i,5}$ with the same color. For each clause unit, $q_2,q_3$, can be labeled with the same color as the $q_1$. This leaves only the $f_1,f_5$ from each literal unit uncolored. The $f_1$ can be colored the same as the $P_{0,2}$ (or $P_{0,5}$ if its a bottom clause unit), but if we colored $P_{0,5}$ ($P_{0,2}$ for bottom) the same as $P_{0,2}$ ($P_{0,5}$ for bottom), then we can't. The same applies for $f_5$, but with the complement literal units (CLUs). However, $f_1$ and $f_5$ can be colored with the same color. Therefore we need only an additional $2n$ colors to color $f_1$ and $f_5$. Therefore, the chromatic number of the complement of the subgraph is $2mn + 6n + m + 1 = k$

Since the chromatic number of the complement of an induced subgraph of the PVG of the construction is equal to $k$, this means that the convex cover of the construction must be at least $k$. Since Shermer shows such a cover, this implies that the convex cover number of the construction is exactly $k$.

Since the convex cover number of the Shermer construction $I'$ for any 3-SAT instance $I$ is equal to $k$, and the hidden set number of $I'$ is strictly less than $k$ if and only if $I$ is unsatisfiable, 3-SAT reduces to deciding if a simple polygon is a homestead polygon or not. Thus, since 3-SAT is NP-hard, deciding if a simple polygon is a homestead polygon or not must also be NP-hard.

\end{proof}

\begin{corollary}
Deciding if a polygon with or without holes is a homestead polygon or not is NP-hard.
\end{corollary}
\begin{proof}
Since deciding if a simple polygon is a homestead polygon is NP-hard, this naturally implies that the more general case of a polygon with or without holes is also NP-hard.
\end{proof}

In addition to these hardness proofs, we present several simple polygons for which $hs(P) < cc(P)$. The first polygon for which it was proven that the hidden set number and convex cover number do not coincide was given in Browne and Chiu \cite{browne2022collapsing}. The polygons we present each have different properties, some of which will be exploited later for approximation algorithms. However, exact algorithms which find a convex cover and hidden set of the same size cannot exist for these subclasses since they are not universally homesteads, meaning that alternative methods would have to be used for both problems than those of the type presented here.

\begin{theorem}
There exists an orthogonal polygon which is not a homestead polygon
\end{theorem}
\begin{proof}

\begin{figure}[ht]
    \centering
\includegraphics[width = 150px]{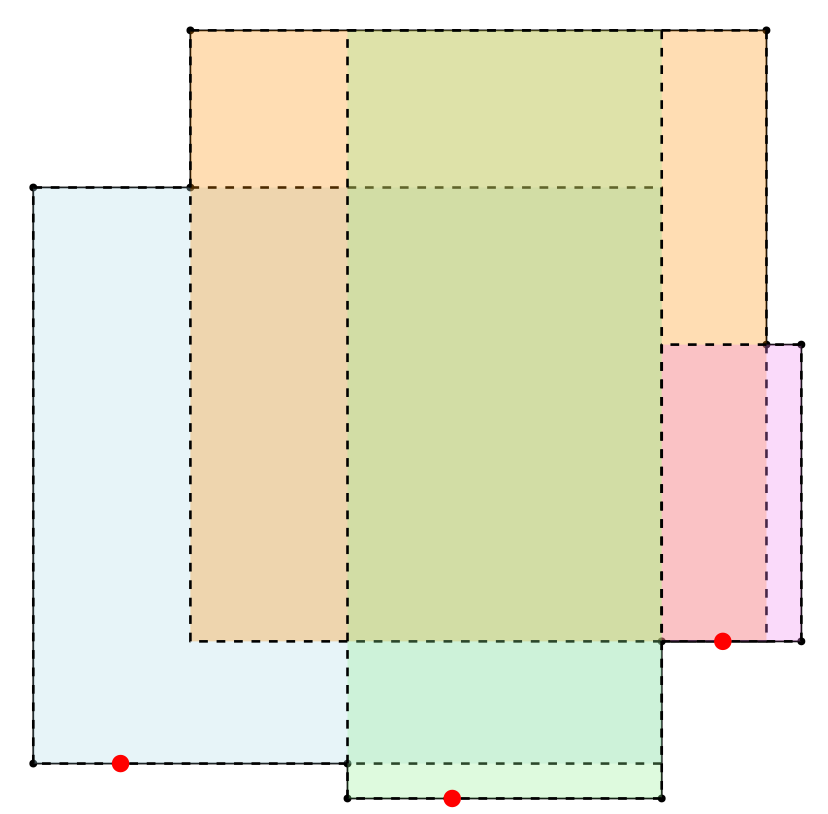}
    \caption{Orthogonal and x-monotone polygon, $M$ which is not a homestead polygon.}
    \label{fig:fig18}
\end{figure}

Observe the example $M$ in Figure \ref{fig:fig18}. This polygon is clearly orthogonal. 
First we will show that $M$ does not admit a hidden set larger than size 3 and then show that the minimum convex cover must be greater than or equal to $4$ by using a theorem from graph theory. 

We start by finding a convex cover of size 4. We can remove out the points that are in the intersection of 2 or more convex pieces since their use can only admit a hidden set of size 3 or smaller. This is because if among the two remaining convex pieces, there can only be an additional two points that can be hidden in them, by definition of convexity. See Figure \ref{fig:fig181}, the regions outlined in purple are the remaining points that have not been eliminated.

\begin{figure}[ht]
    \centering
\includegraphics[width = 200px]{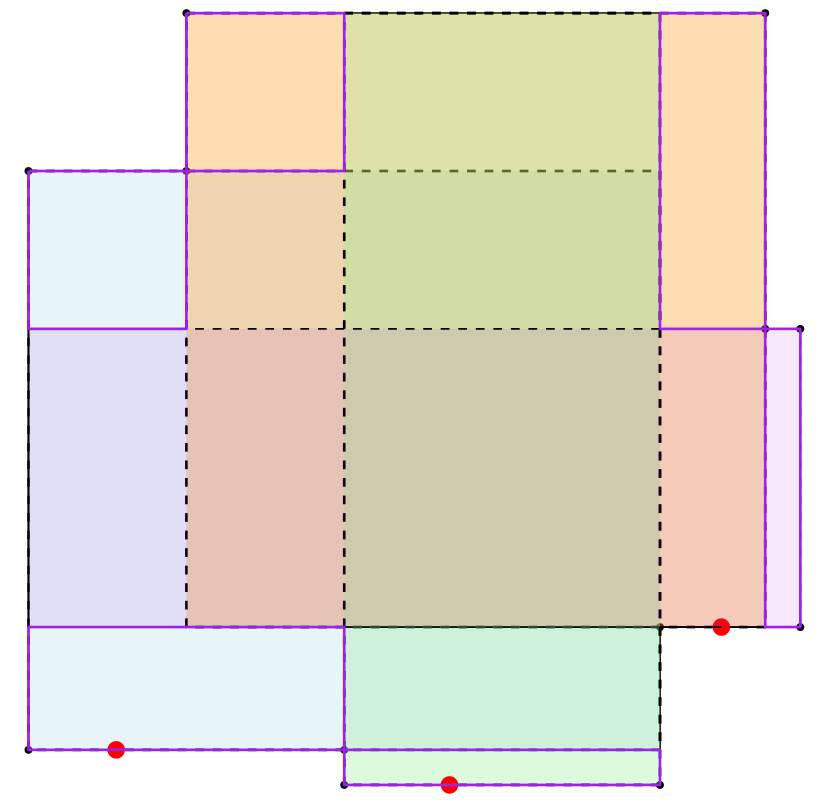}
    \caption{Showing the remaining regions of $M$ that haven't been ruled out by a starting convex cover.}
    \label{fig:fig181}
\end{figure}

Now we show that for each purple region $R_i$, after taking out the strong visibility region of $R_i$, the rest of $M$ can be covered with just 2 convex pieces. By definition of strong visibility, all the points in $R_i$ see all the points in the strong visibility region. Therefore, using any point in $R_i$ disqualifies usage of the strong visibility region, hence covering the rest with 2 convex pieces means at most 2 additional hidden points can be used.

For some of the purple regions, we split $R_i$ into 2 regions and analyze them separately because the strong visibility region of the whole of $R_i$ is just a convex piece and thus we needed 3 to complete that cover, but when the region is split both pieces can have the remaining portion covered with just 2. This full analysis is shown in Figure \ref{fig:fig182}, with the specific region being analyzed in each iteration shaded in purple.

\begin{figure}[ht]
    \centering
\includegraphics[width = 300px]{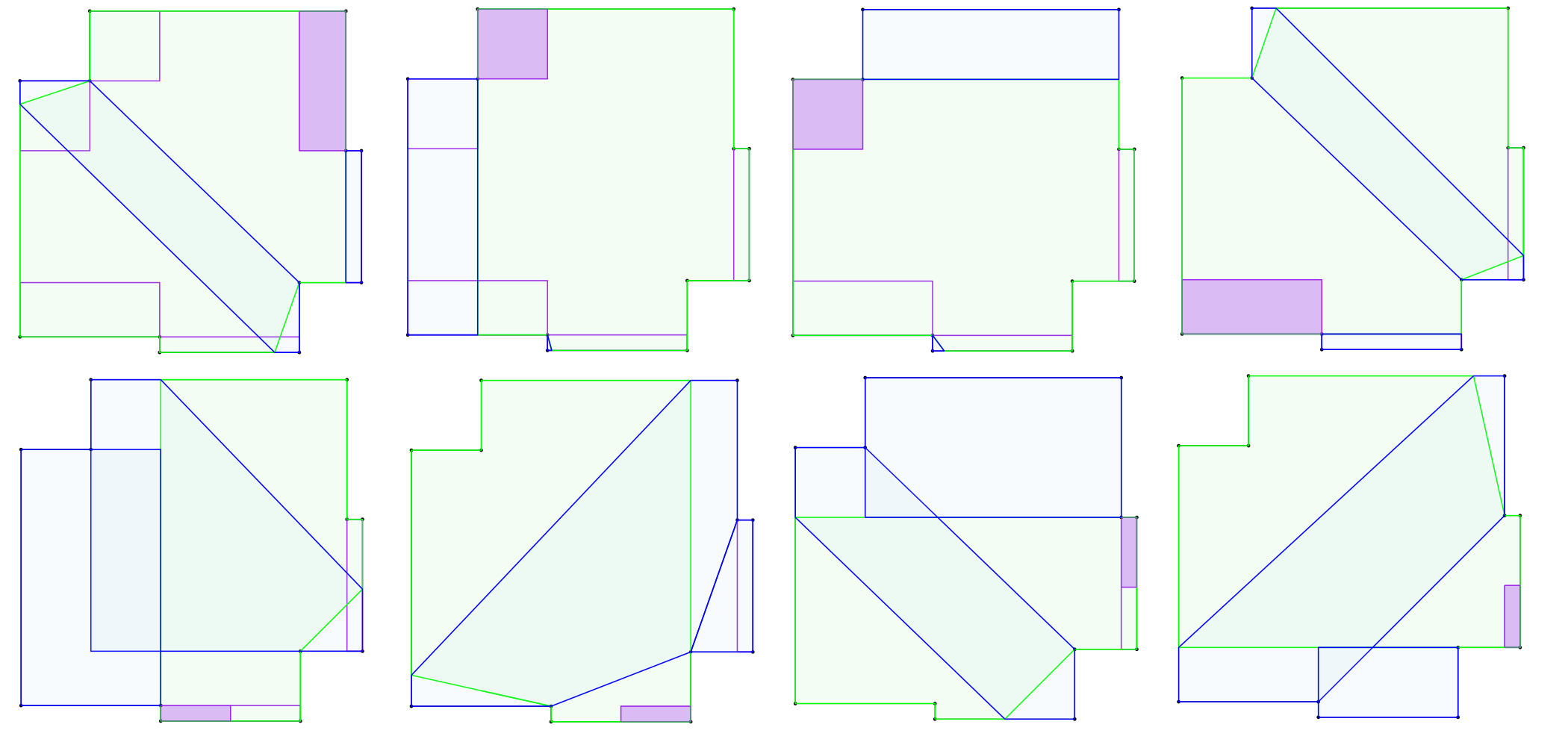}
    \caption{Showing the size 2 convex covers of the remaining region after taking out the strong visibility regions of the purple regions.}
    \label{fig:fig182}
\end{figure}

Since, as is shown in Figure \ref{fig:fig182}, including the strongly visible region of any purple region allows the rest of $M$ to be covered with just 2 added convex pieces, we know that $M$ admits hidden sets of size at most 3. 

For convex cover, we find an indueced subgraph of the PVG of $M$ with clique cover number 4. From Gella and Artes \cite{Gella2014CliqueCO}, we know that for graph $G$ and induced subgraph $H$, the following inequality holds:
\begin{equation}
    cc(G) \geq cc(H)
\end{equation}
Therefore, since convex cover is simply the clique cover of the point visibility graph, finding that clique cover number of the point visbility graph in $M$ restricted to a small set of points is 4 will suffice to show that $cc(M) \geq 4$.

We show this through finding the chromatic number of an isomorphism of the complement graph in Figure \ref{fig:fig183}. Start by observing that nodes 1,2,3 form a clique in the complement, thus must all have different colors, say red, blue and green respectively. Node 6 is adjacent to both red (1) and blue (2), it must be green. Node 5 is adjacent to green (6) and red (1) and thus must be blue. Node 7 is adjacent to green (3) and blue (5) so it must be red. Since 4 is adjacent to green (3), red (7), and blue (5), it must be a fourth color. Therefore the chromatic number of $G'$ must be 4. This in turn proves that $cc(G) = 4$ and in turn $cc(P) \geq 4$. Therefore $M$, an orthogonal polygon, is not a homestead.

\begin{figure}[ht]
    \centering
    \includegraphics[width = 100px]{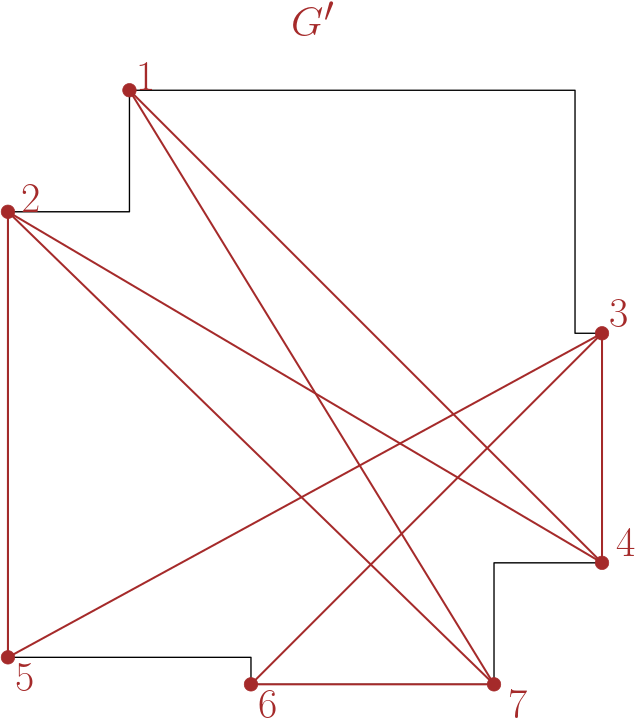}
\includegraphics[width = 100px]{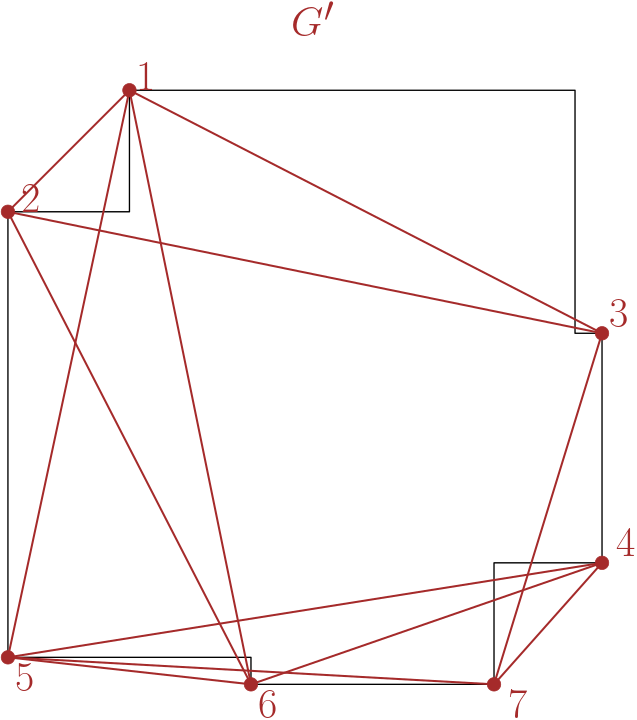}
\includegraphics[width = 100px]{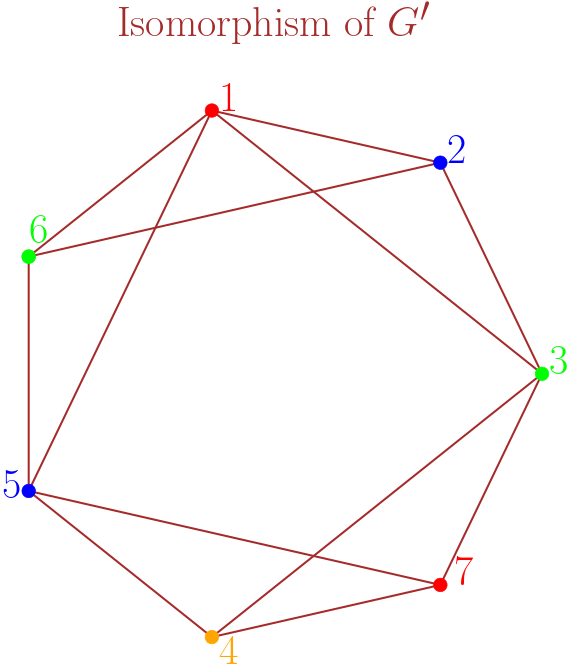}
    \caption{Showing an induced subgraph of $PVG(M)$ with $cc(G) = 4$ from chromatic number of compliment.}
    \label{fig:fig183}
\end{figure}

\end{proof}

\begin{corollary}
There exists an x-monotone polygon which is not a homestead polygon.
\end{corollary}

\begin{proof}
The polygon $M$ is also clearly x-monotone in addition to being orthogonal. Therefore, this corollary is true from the previous proof that $M$ is not a homestead.
\end{proof}

\subsubsection{General position}

Another subclass of polygon that we had conjectured was comprised only of homestead polygons was polygons whose vertices were in general position. A set of vertices is said to be in general position if no 3 vertices are collinear and no 4 vertices are said to be cocircular.

\begin{figure}[ht]
    \centering
\includegraphics[width = 200px]{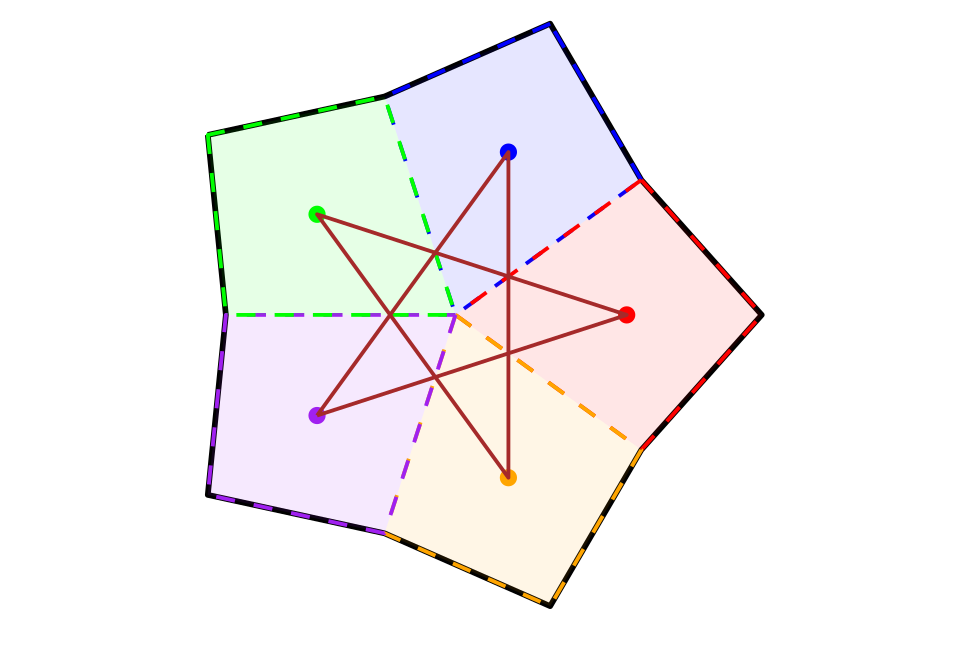}
    \caption{A polygon $W$ which has no three collinear vertices and is not a homestead.}
    \label{fig:fig19}
\end{figure}

\begin{theorem}
There exists a polygon which is not a homestead polygon which has vertices in general position
\end{theorem}
\begin{proof}

First we show a polygon $W$ (in Figure \ref{fig:fig19} which is a polygon where no 3 of the vertices are collinear, but does have cocircular points, and is not a homestead polygon. The polygon exhibits rotational symmetry which allows us to simplify our analysis. First we show that the hidden set number of $W$ is 2, then that its convex cover number is 3 and then give a modification of $W'$ that does not change any of the properties which make $W$ not a homestead.

For hidden set, we analyze the green region of $W$ as, due to rotational symmetry, it is equivalent to analyzing all the other regions. Each of the colored regions have nodes in their centers. These nodes are connected if their regions are strongly visible to each other (their strong visibility regions contain all of the other). Any point in the green region sees all the points in the red and orange regions, meaning we need to cover the blue and purple regions with just one convex polygon. Since blue and purple have an edge, this implies that they see each other, so simply taking their convex hull will produce the convex piece. Because only 1 additional convex piece is needed after using the strong visibility region of the green region, and the green region is, by rotational symmetry, equivalent to all the other regions which total all of $W$, $W$ admits hidden sets of size at most 2.

For convex cover, we take the convex vertices as our subgraph ("points of the star"). It is clear that the graph induced by these vertices is the same as that of the regions from the hidden set discussion, ie the convex vertex in the green region sees the vertices in the red and orange regions but not the blue and purple. This graph is $C_5$, a well known graph which has clique cover number 3.

\begin{figure}[ht]
    \centering
\includegraphics[width = 150px]{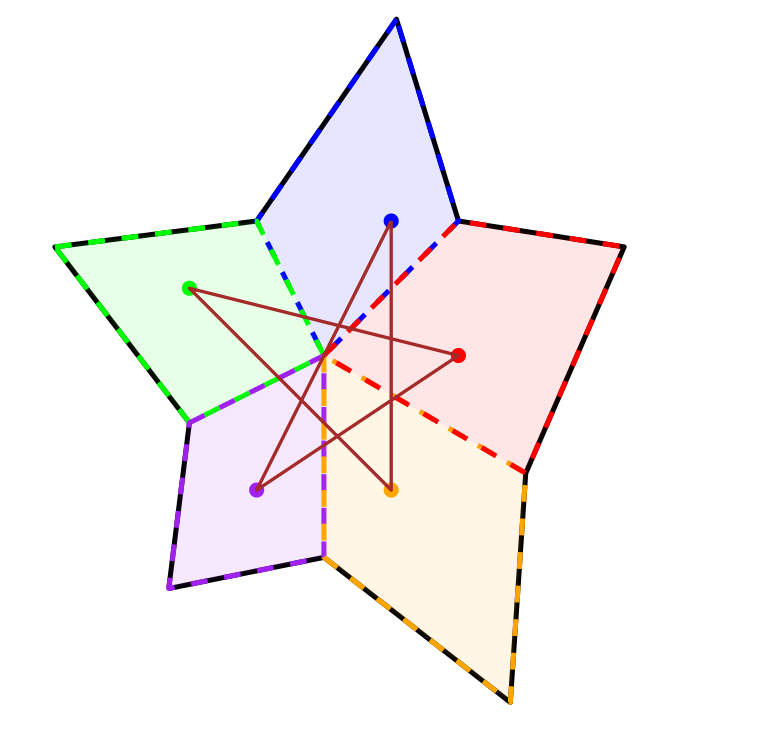}
    \caption{A polygon $W'$ whose vertices are in general position and which is not a homestead polygon.}
    \label{fig:fig20}
\end{figure}

Observing $W'$ (see Figure \ref{fig:fig20}), it is clear that the colored regions are related just as those of $W$. The convex vertices also are related in the same way as those of $W$. Therefore, $W'$, a polygon which has vertices in general position, must have hidden set number of at most 2 and convex convex number of at least 3 and must not be a homestead polygon. 

\end{proof}

\begin{corollary}
There exists a star-shaped polygon which is not a homestead polygon.
\end{corollary}
\begin{proof}
A star shaped polygon is a polygon for which there is some point or set of points that see the entirety of the polygon. The center point of $W$ is in all of the colored regions, therefore must see the entirety of $W$. Therefore, since $W$ is not a homestead polygon, we know that it is not the case that the entire class of star-shaped polygons are homestead polygons.
\end{proof}

\section{Exact and Approximation Algorithms}

\subsection{Polygons with Holes}

For polygons with holes, there is no known approximation algorithm for hidden set which performs better than simply placing a single point inside the polygon, which gives a $O(1/n)$ approximation. However, there does exist a $O(\log n)$-approximation which runs in $O(n^{29}\log n)$ time from Eidenbenz and Widmayer \cite{EidenbenzConvexCover}. Their algorithm starts by showing that a factor of at most 3 is lost by restricting the polygons to unions of cells in a polynomial sized arrangement. Then, they provide a set of maximal convex pieces from which the optimal convex cover restricted to that arrangement must use. From there, the set of maximal convex pieces can be treated as sets and the cells of the arrangement as elements for a set cover. This permits the use of the greedy set cover algorithm on this set of maximal convex pieces, yielding a $O(\log n)$ approximation for the restricted problem, which in turn is only a constant factor away from the actual optimal. In order to make sure this runs in polynomial time, a dynamic program is used to find the greedy choices of convex pieces, since the set of maximal convex pieces is superpolynomial.

If the number of holes is $o(\log n)$, then a better approximation for minimum convex cover can be found using the constant factor approximation for the simple case given by Browne, Kasthurirangan, Mitchell, and Polishchuk \cite{Browne2023}. Simply decompose the polygon with holes into $h+1$ simple polygons by using horizontal chords from the holes to the boundary. This will decompose the polygon into at worst $h+1$ pieces, and further that any convex piece gets split into at most $h+1$ pieces by the splitting chords. Applying the algorithm for simple polygons to each of these pieces yields a $(6h+6)$-approximation. With slightly more care with respect to the chords, a similar approximation for maximum hidden set can be achieved. The main property that we need to maintain is that if two points of the same subpolygon do not see each other within the subpolygon, they don't see each other period. This can be achieved by splitting into exactly $h+1$ pieces using two horizontal chords for each hole: one extending leftward from the leftmost vertex of the hole, and one extending rightward from the rightmost vertex of the hole. Take the largest hidden set from the $h+1$ generated by running the $1/8$-approximation on each subpolygon, which yields a $\frac{1}{8h+8}$-approximation.

\subsection{Simple Polygons}\label{subsec:simple_poly}

Browne, Kasthurirangan, Mitchell, and Polishchuk \cite{Browne2023} gave an argument that by losing only a constant approximation factor (3 for convex cover and 1/4 for hidden set) we can reduce the problem in simple polygons to that of just considering the case of polygons that are weakly visible from a convex edge, which are defined in Section \ref{subsec:weaklyvis}. They achieve this using a window partition algorithm from Suri \cite{suriFirst}, which decomposes a simple polygon into regions by their link distance from a point. See Figure~\ref{staged} for an example of this decomposition When this point is a convex vertex of the polygon, these regions are all polygons weakly visible from a convex edge. 

\begin{figure}[ht]
\centering
\includegraphics[width=.80\columnwidth]{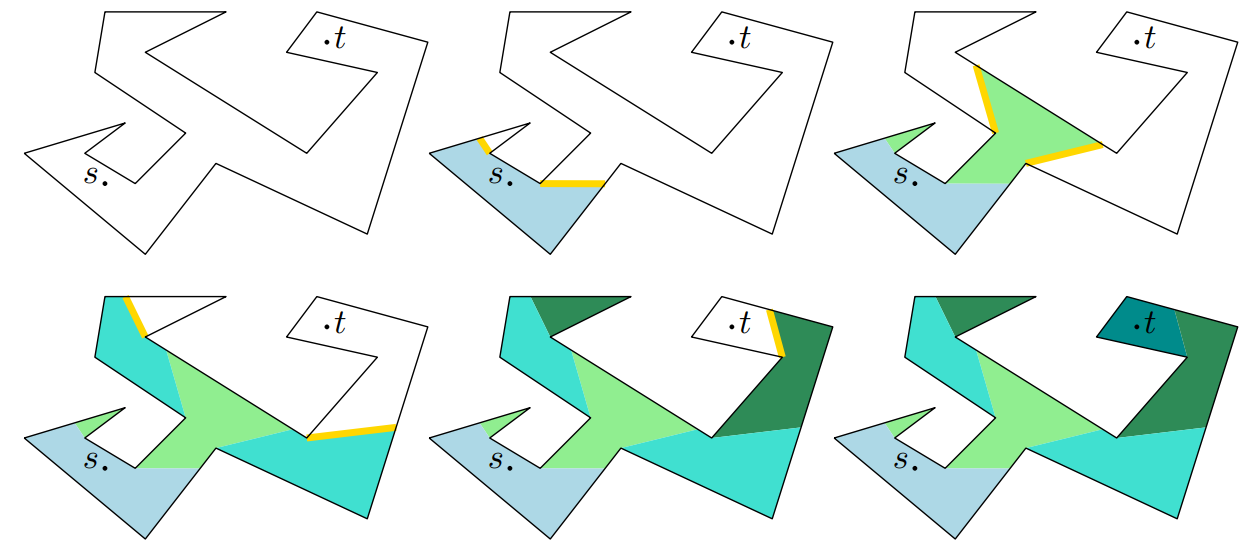}
\vspace{0.5cm}

\includegraphics[page=4,scale=.5]{images/cc.pdf}
\hfill
\caption{Top: Window partition of a simple polygon given in stages (Figure~3 from \cite{revisited}): The windows are yellow. Bottom: Right (red) and left (blue) windows of the weak visibility polygon (gray) of the yellow window. (Figure~6 from \cite{Browne2023})}
\label{staged}
\end{figure}

For convex cover, they argue that no convex body can intersect more than 3 different regions. This means that by restricting convex pieces to stay within one region, only an additional factor of 3 is lost. Using their approximation for convex cover in polygons weakly visible from a convex edge gives a 2-approximation for the convex cover of each region. Taking the union of these covers gives a cover of the whole polygon which is at worst 6 times as bad as the optimal solution.

For hidden set, they use an argument from Alegr{\'i}a, Bhattacharya and Ghosh \cite{Alegra2019A1A} that partitions the set of regions into 4 subsets. They define the subset membership by the link distance modulo 2 and the type of turn which was made from the last window to the window illuminating this region, i.e. a left turn or a right turn. We give their figure for distinguishing between left and right windows in Figure~\ref{staged}. These 4 subsets are each visually independent sets of regions, meaning that if you take all the regions within a subset, two points can only see each other if they are in the same region. Therefore, if we have hidden sets in each region, for each subset we can take the union of the hidden sets for its members since the regions are visually independent of each other. If we take the largest of these 4 hidden sets, then that would give a 1/4 approximation. However, their algorithm gives a 1/2 approximation for the case of polygons weakly visible from a convex edge, so these factors multiply to 1/8.

\subsection{Polygons Weakly Visible from a Convex Edge}\label{subsec:weaklyvis}

Browne, Kasthurirangan, Mitchell, and Polishchuk \cite{Browne2023} give a 2-approximation for minimum convex cover and a 1/2 approximation for maximum hidden set. They do so by finding a hidden set and a convex cover such that the convex cover is at most 2 times the size of the hidden set, which implies those factors due to the inequality from Shermer \cite{shermer_1989}, $hs(P) \leq cc(P)$. Their base algorithms for both problems run in $O(n^{2+o(1)})$ time, based on recent results on maximum flow \cite{ChenUpdated}. However, they remark that the runtime can be decreased to $O(n^2)$ for maximum hidden set, the details of which we present here in Section~\ref{sec:no_antichain}.

\subsubsection{Using Edges as a Poset}

The key insight of Browne, Kasthurirangan, Mitchell, and Polishchuk \cite{Browne2023} is that we can treat the edges of a polygon weakly visible from a convex edge as a partially ordered set (poset). This poset is determined by strong visibility, i.e. whether all of the points in an edge $e_i$ see all of the points in a edge $e_j$, and then if strong visibility exists the elements are comparable, with the comparison being the counterclockwise order of the edges. This can also be thought of as a DAG, where there is an arc pointed towards the later of the two edges. Importantly, this relation is transitive, so if for $i<j<k$, $e_i$ strongly sees $e_j$ and $e_j$ strongly sees $e_k$, then $e_i$ strongly sees $e_k$.

\begin{figure}[ht]
\centering
\includegraphics[page=1,width=0.5\columnwidth]{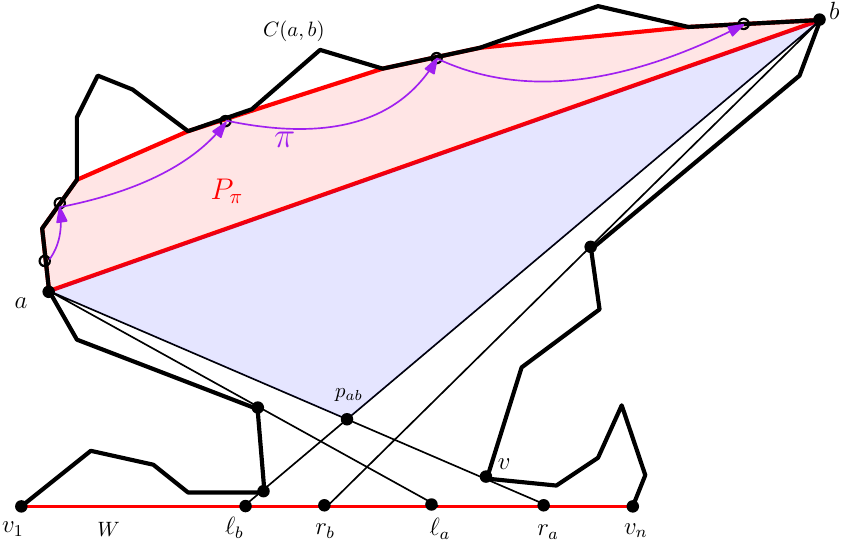}
\caption{Figure 1 from \cite{Browne2023}. Arcs indicating the partial ordering of selected edges.}\hfill
\label{fig:wvp}
\end{figure}

Using this poset, we can find a set of edges for which no pair strongly see each other, also called an \textit{antichain}. These edges may share partial visibility, so we need to pick special points along each edge so that they form a hidden set. Specifically, they choose points that are sufficiently close to one of the two endpoints of the edge based on a variation of the visbility decomposition of the polygon. From there, it suffices to simply choose the correct endpoint to be near for each edge, which is done using a recursive procedure on the edges. We give an alternate set of these candidates in Section~\ref{sec:no_antichain}.

For convex cover, they start by considering the minimum path cover of the poset. From this path cover, they construct a set of convex pieces which partially cover the polygon. From there, they augment the convex pieces to include some additional area. After this augmentation, they add an additional convex piece for every one of the original pieces, doubling the number of pieces. These additional convex pieces prove to be sufficient to complete the cover.

Due to Dilworth's \cite{dilworth1950decomposition} theorem, the maximum antichain and the minimum path cover of a poset are of the same size. This means there are at most twice as many convex pieces as hidden points, demonstrating their factors of 2 and 1/2 respectively. Due to recent advances in the maximum flow problem \cite{ChenUpdated}, finding either the maximum antichain or the minimum path cover both take $O(n^{2+o(1)})$ deterministically.

\subsubsection{A hidden set without an antichain}\label{sec:no_antichain}

We give the details of the improved algorithm for the maximum hidden set which runs in $O(n^2)$ mentioned in Browne, Kasthurirangan, Mitchell and Polishchuk \cite{Browne2023}.

\begin{lemma} \label{lem:superset_of_hidden}    
There exists a set $S\subset \partial P$ of $2n-2$ points such that for any antichain $I$ in $G$, there is a hidden subset of points $H\subset S$ with $|H| = |I|$. Such a set $S$ can be computed in polynomial time.
\end{lemma}

\begin{proof}
First we describe how to obtain the set $S$, and then prove its properties. 
For each edge $e_i$, $i=1,2,\ldots,n-1$, we find two points $s_i\in e_i$ (near $v_i$) and $t_i\in e_i$ (near $v_{i+1}$); these $2n-2$ points constitute the set~$S$.

We compute the shortest path trees, $SPT(v_i)$, for each vertex $v_i$; this takes $O(n)$ time per vertex~\cite{DBLP:journals/algorithmica/GuibasHLST87}, so overall time $O(n^2)$. This allows us to access the parent of $v_j$ with respect to $SPT(v_i)$ in $O(1)$. 

This in turn allows us to determine whether the $SP(v_i,v_j)$ contains no vertex (besides $v_i,v_j$), or $\geq 1$ vertices in $O(1)$ time. We can use this to calculate $s_i$ and $t_i$ ie. how close each must be to $v_i$ and $v_{i+1}$, respectively. Each edge pair $e_i, e_j$ defines 2 constraints, one on either $s_i$ or $t_i$ and one on either $s_j$ or $t_j$. We find these 2 constraints simultaneously. Assume without loss of generality, that $i < j$. This assumption allows us to know that we are constraining $t_i$ and $s_j$ and to consider $SP(v_{i+1},v_j)$. There are 2 cases depending on the complexity of $SP(v_{i+1},v_j)$. 

\begin{enumerate}
    \item[Case (1)] $SP(v_{i+1},v_j)$ is a line segment. If this is the case, then we will consider the extensions of both $e_i$ and $e_j$. If the extension of $e_j$ intersects the edge $e_i$ at a point $p_i$, then $t_i$ must be strictly closer to $v_{i+1u}$ than $p_i$ is to $v_{i+1}$. If the extension of $e_i$ intersects the edge $e_j$ at a point $p_j$, then $s_j$ must be strictly closer to $v_j$ than $p_j$ is to $v_j$. Otherwise, there is no constraint made. See the left image in Figure \ref{fig:hidden_set_dp}.
    
    \item[Case (2)] $SP(v_{i+1},v_j)$ contains at least one vertex (other than $v_{i+1}, v_j$). Let us call this vertex $u$. Consider the line parallel to $\overleftrightarrow{v_{i+1} v_j}$ that passes through $u$. If this line intersects the edge $e_i$ at a point $p_i$, then $t_i$ must be strictly closer to $v_{i+1}$ than $p$ is to $v_{i+1}$. If this line intersects the edge $e_j$ then $s_j$ must be strictly closer to $v_j$ than $p_j$ is to $v_j$. Otherwise, there is no constraint made. See the right image in Figure \ref{fig:hidden_set_dp}.
    
\end{enumerate}

\begin{figure}[ht]
    \centering
    \includegraphics[width=.6\textwidth]{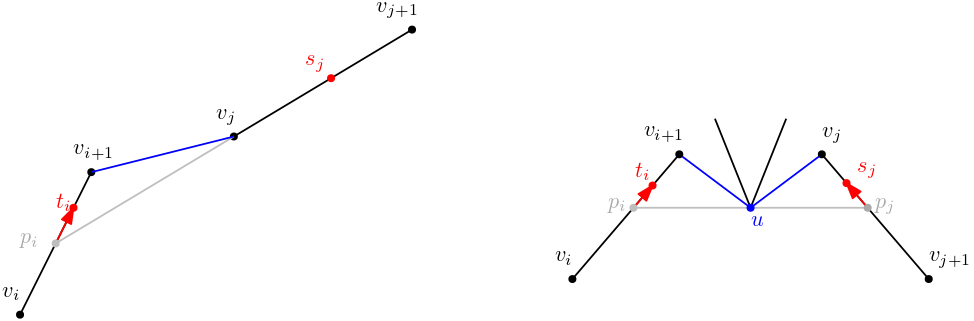}
    \caption{Cases for placements of $s_i$ and $t_j$ for edges $e_i = v_iv_{i+1}$ and $e_j= v_jv_{j+1}$ in $P$. Left: Case (1), right: Case (2). In considering the hidden subset: $i = i_{1}, j = i_{|I|}$.}
    \label{fig:hidden_set_dp}
\end{figure}


After completing all these pairwise comparisons, we take the most restrictive constraint for each $s_i$ (closest to $v_{i+1}$) and $t_i$ (closest to $v_i$). Set the point to halfway between this constraint and the respective vertex (or the midpoint of the edge if there is no constraint). Each constraint takes $O(1)$ to compute since we can use the SPT parent of either as the one vertex in Case (2). There are $\leq n$ constraints per point ($s_i$ and $t_i$), which means it takes a total of $O(n^2)$ time to compute $S$.

Now we show that such a set $S$ admits a hidden subset $H$ for every antichain $I$ such that $|H| = |I|$. Particularly, there are three conditions (a),(b),(c) such that there is at least one $H$ that satisfies both (a),(b) and one $H$ which satisfies both (a),(c). These can be the same subset satisfying all three. (a) For each edge $e_i \in L$, $H$ admits either $s_i$ or $t_i$. (b) The edge of lowest index $e_{i_{1}}$ admits $t_{i_{1}}$ into $H$, and (c) The edge with the highest index $e_{i_{|I|}}$ admits $s_{i_{|I|}}$ into $H$. 

We can show this by induction on $|I|$. If $|I| = 1$, then $e_{i_{1}} = e_{i_{|I|}}$. Therefore, to satisfy (a)(b) we use $\{t_{i_{1}}\}$. To satisfy (a)(c) we use $\{ s_{i_{|I|}} \} 
 = \{s_{i_{1}} \}$. Since the base case holds, we then show that if this holds for $|I| \in [1, m-1]$, then it also holds for $|I| = m$. 

For $|I| = m$, we have the same 2 cases as in the construction of $S$, but here we work with $e_{i_{1}}$ and $e_{i_{|I|}}$ instead of $e_{i}$ and $e_{j}$. Refer to Fig.~\ref{fig:hidden_set_dp} for a visual of the cases.

\begin{enumerate}
    \item[Case (1)] $SP(v_{i_{1}+1},v_{i_{|I|}})$ is a line segment. Therefore $v_{i_{1}+1}$ sees $v_{i_{|I|}}$. For $e_{i_{1}},e_{i_{|I|}}$ to not have an arc in $G$, either $v_{i_{1}}$ or $v_{i_{|I|+1}}$ lies above the line $\overleftrightarrow{v_{i_{1}+1},v_{i_{|I|}}}$. Otherwise, by chord property they would see each other. If this is the case, then the only way that $t_{i_{1}}$ could see $s_{i_{|I|}}$ is if one of these lies below the edge extension of whichever lies above $v_{i_{1}+1},v_{i_{|I|}}$, which we know is not the case from the constraints given in the construction above. See the left image of Fig.~\ref{fig:hidden_set_dp}. Since one must lie above the line $\overleftrightarrow{v_{i_{1}+1},v_{i_{|I|}}}$, either one or neither can see into the line segment $v_{i_{1}+1},v_{i_{|I|}}$.
    
    Consider the case where $t_{i_{1}}$ sees into ${v_{i_{1}+1},v_{i_{|I|}}}$. From the inductive assumption, we know there is a subset $H_{i_{1}}$ corresponding to $I_{i_{1}} =  I \setminus e_{i_{|I|}}$ that satisfies (a),(b). Since $e_{i_{1}}$ has a lower index than all other edges in $I_{i_{1}}$, we know that $t_{i_{1}}$ is admitted into $H_{i_{1}}$. Since $s_{i_{|I|}}$ cannot see into ${v_{i_{1}+1},v_{i_{|I|}}}$, it cannot see any other edge in $I_{i_{1}}$. Because $s_{i_{|I|}}$ has a higher index than all edges in $I_{i_{1}}$, $H_{i_{1}} \cup s_{i_{|I|}}$ forms a hidden subset of size $|I|=m$ that satisfies all three of (a),(b),(c). If instead $s_{i_{|I|}}$ sees into ${v_{i_{1}+1},v_{i_{|I|}}}$ or neither do, the argument is identical, except that $H_{i_{|I|}}$ from $I_{i_{|I|}} = I \setminus e_{i_{1}}$ will satisfy (a),(c) and then the union $H = H_{i_{|I|}} \cup \{ t_{i_{1}} \}$ will satisfy (a),(b),(c).

    \item[Case (2)] $SP(v_{i_{1}+1},v_{i_{|I|}})$ contains at least one vertex (other than $v_{i_{1}+1},v_{i_{|I|}}$). From the comparison between $e_{i_{1}}$ and $e_{i_{|I|}}$, we know that both $t_{i_{1}}, s_{i_{|I|}}$ lie to the left of the directed line parallel to $\overleftrightarrow{v_{i_{1}},v_{i_{|I|}}}$ through some vertex $u$ of $SP(v_{i_{1}+1},v_{i_{|I|}})$. This is because $SP(v_{i_{1}+1},v_{i_{|I|}})$ makes only left turns \cite[Lemma~2]{GhoshWeakVisible}, meaning $u$ must lie to the right of directed line $\overleftrightarrow{v_{i_{1}},v_{i_{|I|}}}$. $t_{i_{1}}$ and $s_{i_{|I|}}$ must both lie between these two lines. This further means no point in the interval $[t_{i_{1}},u)$ can see any point in the interval $(u, s_{i_{|I|}}]$, as this would violate the simplicity of the polygon. See the right image of Fig.~\ref{fig:hidden_set_dp}.
    
    This allows us to consider all the edges of $I$ in  interval $[t_{i_{1}},u)$, which we denote as $I_{i_{1}}$, and all the edges of $I$ in the interval  $[u,s_{i_{|I|}})$, which we denote as $I_{i_{|I|}}$. We know a hidden subset $H_{i_{1}}$ abiding by (a),(b) exists for $I_{i_{1}}$ and a hidden subset $H_{i_{|I|}}$ abiding by (a),(c) exists for $I_{i_{|I|}}$. Because no point in one interval sees any point in the other, $H_{i_{1}} \cup H_{i_{|I|}}$ must be a hidden subset of size $|I|=m$, abiding by all three (a),(b),(c).

\end{enumerate}

Therefore, since in all cases we can obtain a hidden subset $H$ of $S$ such that $|H| = |I|$, we know, by induction, that Lemma \ref{lem:superset_of_hidden} holds. We also know that such a set $S$ can be calculated in $O(n^2)$ time.
\end{proof}

\begin{theorem} \label{thm:hidden_set_dp}
    For a \wvp $P$, a hidden set of size greater than or equal to the size of the largest antichain in $G$ can be found in polynomial time.
\end{theorem}
\begin{proof}
    From Lemma \ref{lem:superset_of_hidden}, we know that there exists a set of points $S$ of size $2n-2$ where for any antichain we can take a subset of $S$ as a hidden set of the same size as the antichain. This means that for the largest antichain, there exists a hidden subset of $S$ of the same size.

    Because the set of points $S$ are on the boundary on $P$, this means that we can define a degenerate polygon $P'$ with $3n-2$ vertices that includes the points $S$ as additional vertices to those in $P$. We include these points by simply splitting the edges that they lie on, changing nothing about the actual geometry of $P$. Because none of the points in $S$ are along $\W$, $P'$ remains a \wvp. We can apply the $O(n^2)$ dynamic programming (DP) algorithm from \cite{GhoshWeakVisible} to find a maximum hidden vertex set among these points. To be specific, we will exclude the vertices $v_1$ and $v_n$ from consideration in placing hidden vertices. Excluding the vertices of $\W$ is important for extending this to arbitrary simple polygons. Because the DP on $P'$ considers all points in $S$, this hidden vertex set must be at least as large as the hidden subset of $S$ corresponding to the largest antichain. 

    The total runtime of this reduces to simply finding $S$ and then applying the algorithm from \cite{GhoshWeakVisible}, both of which are $O(n^2)$. From earlier, we know that a hidden set that is the same size as the largest antichain of the edges is a 1/2-approximation for the maximum hidden set, so this provides a faster approach to finding an approximate maximum hidden set for \wvps and simple polygons in general by extension.    
\end{proof}

\subsection{Convex cover and hidden set in convex fans and monotone mountains}

A convex fan is a polygon $P$ which contains a convex vertex $v$ where for all points $p\in P$, $v$ sees $p$. Here, convex vertices can also include straight vertices, those with angle measure exactly 180$^\circ$.

\begin{theorem}\label{thm:confan}
    For any convex fan $P$ with $n$ vertices, a maximum hidden set and a minimum convex cover of $P$ can be found in $O(n^{2+o(1)})$ time. 
\end{theorem}

\begin{proof}
     We know that for either edge incident on $v$, this edge must be convex. Otherwise, $v$ would not see the points on the boundary after the reflex vertex, which contradicts the definition of a convex fan. This means that $P$ is weakly visible from a convex edge ($e_i$), and the edges of $P$ admit a partial ordering from the strong visibility relation from Browne, Kasthurirangan, Mitchell, and Polishchuk (see Section \ref{subsec:weaklyvis}).

Since each edge $e_i$ in $P$ strongly sees $v$, it must be the case that the $CH(e_i \bigcup v) \subset P$, where $CH$ indicates the convex hull. We can produce a decomposition of $P$ using these convex pieces. It also is the case that every set of edges that strongly see each other will also strongly see $v$, which means that for every chain in the poset, we can construct a convex piece that includes $v$. These convex pieces are a superset of the pieces in the decomposition corresponding to the edges in the chain, which means that they form a complete convex cover. This convex cover is equivalent to the size of the hidden set found from the anti-chain using the approach from Browne, Kasthurirangan, Mitchell, and Polishchuk (see Section~\ref{subsec:weaklyvis}). We give an example of such a set and convex cover in the right subfigure of Figure~\ref{fig:monotonefan}. This procedure takes time equivalent to the case of \wvps, $O(n^{2+o(1)})$.
\end{proof}

\begin{figure}[ht]
    \centering
    \includegraphics[width=.7\textwidth]{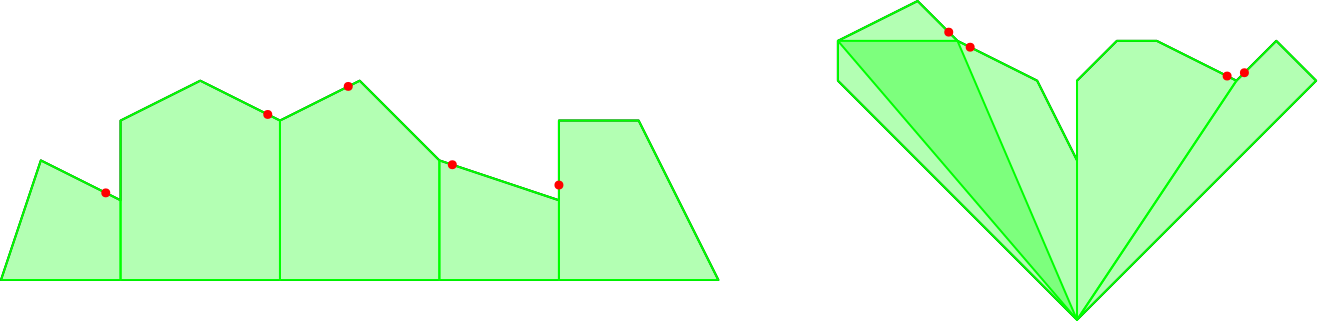}
    \caption{Showing optimal hidden sets and convex covers for a monotone mountain (Left) and a convex fan (Right)}
    \label{fig:monotonefan}
\end{figure}

Monotone mountains are simple polygons composed of an $x$-monotone chain and a single edge \W formed by the vertex with the lowest $x$-coordinate and the vertex with the highest $x$-coordinate. An $x$-monotone chain is a polygonal chain composed of vertices whose $x$-coordinates are non-decreasing. 

\begin{theorem}\label{thm:monmount}
    For any monotone mountain $P$ with $n$ vertices, a maximum hidden set and a minimum convex cover of $P$ can be found in $O(n^{2+o(1)})$ time. 
\end{theorem}

\begin{proof}
    A similar argument as in Theorem~\ref{thm:confan} can be made for the case of monotone mountains. 
    
    Since the opposite chain (\W) of the $x$-monotone chain is simply an edge, every point along the boundary sees the corresponding point on \W with the same $x$-coordinate, since a piece of boundary would have to violate $x$-monotonicity to block this vision. This means that \W is a convex edge which weakly sees the entirety of $P$.

    It is also the case that if an edge $e_i$ strongly sees another edge $e_j$, $e_i$ must also strongly see the corresponding points to $e_j$ along \W. This follows from the $x$-monotocity as well as the chord property \cite{GhoshWeakVisible}[Lemma 1]. By the chord property, we know that no piece of boundary between the edges in clockwise order blocks $e_i$ from seeing the corresponding points of $e_j$ along \W. Therefore, either a piece of boundary from before $e_i$ must block it or a piece of boundary after $e_j$, which would violate $x$-monotonicity.

    Since the strong visibility relation between two edges also applies to the regions below these edges down to the corresponding points in \W, this allows us to decompose the polygon into these vertical "slabs", just like with convex fans. The union of these slabs is equal to $P$. Therefore, we can again use the chain cover of the edges to obtain a set of convex pieces that will now cover the entirety of $P$. The hidden set algorithm remains exactly the same. Therefore, since we have a hidden set and convex cover of the same size, these must both be optimal. This procedure takes time equivalent to the case of \wvps, $O(n^{2+o(1)})$. We give an example convex cover and hidden set in Figure~\ref{fig:monotonefan}.
\end{proof}

\subsection{Rocking chair polygons}

We give a generalization of monotone mountains which we refer to as rocking chair polygons. We define a rocking chair polygon $P$ as a polygon composed of a $x$-monotone chain $v_1, v_2, . . . v_t$ and a convex chain $v_t, v_{t+1}, ... v_1$. Without loss of generality, we can consider $v_1$ to be the point with the lowest $x$-coordinate and $v_t$ to be the point with the highest $x$-coordinate and that the vertices are in clockwise order.

\begin{theorem}
    For any rocking chair polygon $P$,  a maximum hidden set and a minimum convex cover of $P$ can be found in $O(n^{2+o(1)})$ time. 
\end{theorem}

\begin{proof}
    This is simply a generalization of Theorem \ref{thm:monmount}. Consider the point with the lowest $y$-coordinate, $v_{min}$. We can draw a horizontal line through $v_{min}$ and consider the monotone mountain formed by extending vertical segments down from $v_1$ and $v_t$ onto this horizontal. This forms a monotone mountain $P'$. This monotone mountain can be covered with convex pieces using the algorithm from Theorem \ref{thm:monmount} and have hidden points placed similarly. Specifically, we will place the points on the x-monotone chain, so exclude the two vertical segments from consideration. Since the monotone mountain coverage algorithm from Theorem \ref{thm:monmount} clearly covers these vertical segments, we will still have a cover of $P'$ and a hidden set of the same size. While the hidden points of $P'$ are also hidden points in $P$, the convex pieces will not necessarily be constrained to $P$, such as in Figure \ref{fig:rocking_chair}. However, due to the convexity of the convex chain, if we consider $C = C' \cap P$ for every convex piece $C'$ of the convex cover of $P'$, all pieces $C$ will be convex. Since $P \subset P'$, $\bigcup_{C' \in cc(P')} (C \cap P) = P$. 

    \begin{figure}[ht]
        \centering
        \includegraphics[width=.5\textwidth]{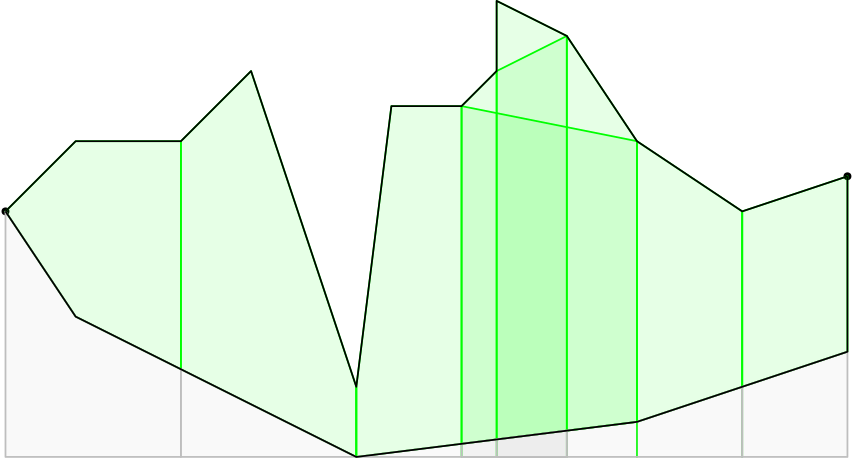}
        \caption{A convex cover and a hidden set in a rocking chair polygon (green and red) and the superimposed monotone mountain (grey)}
        \label{fig:rocking_chair}
    \end{figure}

\end{proof}

\subsection{Convex cover and hidden set in star-shaped polygons}
A star shaped polygon is one which has a nonempty kernel. The kernel of a polygon $P$ is the set of points $K$ for which every point $p\in P$ sees \textit{all} points in $K$. 
\begin{theorem}
    A 2-approximate convex cover and a 1/2-approximate hidden set in a star-shaped polygon can be found in polynomial time.
\end{theorem}

\begin{proof}\label{thm:star}
    We can use the results from Theorem~\ref{thm:confan} to accomplish this. The kernel of a star shaped polygon can be found in $O(n)$ time \cite{linearKernel}, thus allowing us to find a point $p$ in $P$ which sees every other point in $P$. We can consider an arbitary convex vertex $v$ of $P$ and extend the line segment $vp$ until it forms a chord of $P$. This chord partitions $P$ into two (degenerate) convex fans. We can use the algorithm from Theorem~\ref{thm:confan} on both of these, which runs in $O(n^{2+o(1)})$ time. We take the larger of the two hidden sets and the union of the convex pieces, which yield a 2-approximation for convex cover and a 1/2 approximation for hidden set. 
\end{proof}

\subsection{Convex cover and hidden set in weak visibility polygons}

A weak visibility polygon is a polygon $P$ for which a chord of $P$ weakly sees all of $P$. This is a relaxation of the requirements for a \wvp. 

\begin{theorem}
    A 4-approximate convex cover and a 1/4-approximate hidden set in a weak visibility polygon can be found in polynomial time.
\end{theorem}
\begin{proof}
    This can be done nearly identical to Theorem~\ref{thm:star}. Ghosh et al. \cite{GhoshWeakVisible} provide an algorithm that finds a chord which sees all of $P$ (if it exists) that runs in $O(n^2)$ time. This allows us to split $P$ into two subpolygons that are \wvps. This means we can apply the algorithm from Browne, Kasthurirangan, Mitchell, and Polishchuk\cite{Browne2023} to achieve a 4-approximation (and a 1/4 approximation), again taking the better of the two hidden sets (1/2 approximate) and all the convex pieces. 
\end{proof}

\subsection{Convex cover and hidden set in simple orthogonal polygons}

Simple orthogonal polygons are simple polygons for which all angles are either $90^\circ$ or $270^\circ$. We give a 2-approximation for convex cover and a 1/2-approximation for hidden set.

\begin{theorem}
    A 2-approximate convex cover and a 1/2-approximate hidden set in a simple orthogonal polygon can be found in polynomial time.
\end{theorem}
\begin{proof}
    
Franzblau \cite{Franzblau} gave an $O(n \log n)$ algorithm which finds a 2-approximation for finding a minimum number of axis-aligned rectangles to cover a simple orthogonal polygon. Since axis-aligned rectangles are convex, this also yields a convex cover. This algorithm begins by extending all horizontal edges to form a partition of the polygon into axis-parallel rectangles. Although unnecessary for improving the 2-factor, the algorithm then extends these rectangles maximally in both vertical directions and removes any rectangle which are made redundant i.e. are completely covered by other rectangles in the set.

For simple orthogonal polygons, we can find a hidden set of size at worst 1/2 the size of this. We simply place two sets of hidden points on all non-collinear horizontal edges. In the case of collinearities within a set, just choose 1 from each group of pairwise collinear edges. Ignore collinear edges that do not see each other. The two sets comprise of one for the ``top edges'' i.e. those which bound $P$ from above and ``bottom edge'' that bound from below. If we place the hidden points at the midpoints of each edge, we know that no two points on top edges will see each other because the line segment between them must pass through the outside of $P$ near whichever of the two has a lower $y$-coordinate. By symmetry, the same applies to the bottom edges. We choose whichever of the two hidden sets is larger.

We show that there is at most $2 hs(P) - 1$ pieces in the convex partition of size $p$, specifically one less than the total number of (noncollinear) top edges and bottom edges. Let us say that the total number of these horizontal edges is $e$, so $p \leq e-1$. We can show this by induction on $p$. This is clearly holds for when $p=1$, as a rectangle with two edges and one piece, $1 \leq 2-1 = 1$. For the inductive step, consider the lowest rectangle in a polygon $P$ with $k+1$ rectangles in its partition. Consider the polygon(s) left when the lowest rectangle is removed. We will call this $P'$ if it is a single polygon, or $P_1', P_2', . . ., P_m'$ if there are $m$ polygons left. There is two cases for this. In the first case, the top edge of the removed rectangle is a proper subset of a lowest bottom edge of $P'$. See the left subfigure of Figure~\ref{fig:orthog_hidden}. We know that for $P'$, $p' \leq e'-1$. We know that $P$ has an additional piece and an additional bottom edge, which means that $p-1 \leq e-1-1$ which in turn means $p \leq e-1$.

\begin{figure}[ht]
    \centering
    \includegraphics[width=.7 \textwidth]{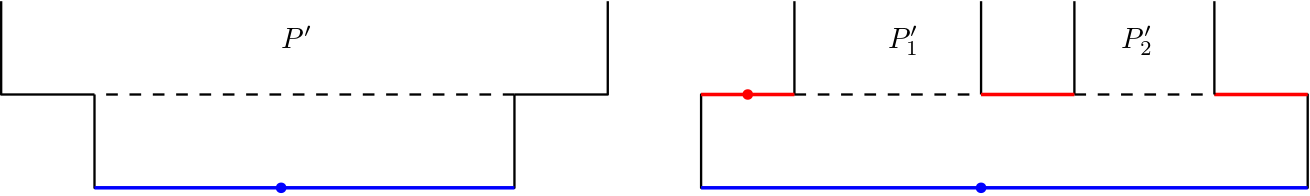}
    \caption{Showing the two cases for the induction. Left: Case 1, the top of the lowest rectangle is a subset of a lowest edge of $P'$. Right: Case 2, the top of the lowest rectangle is a superset of some lowest edges of $P_1' , P_2', . . . , P_m'$}
    \label{fig:orthog_hidden}
\end{figure}

In the second case (see the right subfigure in Figure~\ref{fig:orthog_hidden}), the top edge of the removed rectangle is a proper superset of some nonempty set of the lowest edges from $P_1', P_2', ..., P_m'$. We know that each of $P_i'$ only shares 1 contiguous interval with the top of the removed rectangle, since otherwise $P$ would have a hole. We know that for each $P_i'$, $p_i' \leq e_i'-1$. We know that $P$ contains 2 (noncollinear) edges not in any $P_i'$, a top edge and a bottom edge. The top edge of the removed rectangle cannot be fully covered by the $P_i'$s, since then there would be no splitting of $P$ to create the removed rectangle. Each $P_i$ contains one edge that is along the top edge of the removed rectangle (thus is not a true edge of $P$), which means that $e = 2+ \sum (e_i'-1)$, and the partition includes all, so $p = 1+\sum p_i'$. From the previous inequality, we have:

$$1 + \sum p_i' \leq 1+ \sum (e_i'-1)$$
$$p \leq e-1$$

Therefore, since this holds for both cases, the claim holds by induction.

Therefore, simply taking the larger of the two sets of hidden points yields a 1/2 approximate solution for maximum hidden set and the set of rectangles form a 2-approximation for minimum convex cover. This takes a total of $O(n \log n)$ time for both, as the hidden set can simply be implemented by sorting the horizontal edges by $y$-coordinate and placing candidate hidden points on each edge that is not collinear with its predecessor.

\end{proof}

\subsection{Funnel Polygons}\label{sec:funnel}

These results and those provided in Section \ref{sec:pseudtri} were originally presented in an Arxiv submission by Browne \cite{browne2023convex}. We omit the pseudocode given there, as well as the results on the restricted variants.

\textit{Funnel polygons} are defined as simple polygons composed of a convex edge $e_n = v_nv_1$ with two reflex chains $R_1 = v_1,v_2,...,v_t$ and $R_2 = v_t,...,v_{n-1},v_n$. Assume without loss of generality that $t \geq n-t$, i.e. $|R_1| \geq |R_2|$. We will also assume that each reflex vertex is strictly reflex, i.e. no three vertices within a chain are collinear and there is a turn at every vertex. If this is not the case, we can simply remove the vertices that are not strictly reflex. Funnel polygons are also sometimes referred to as \textit{tower polygons}.

We will use a recursive procedure to compute the hidden set and convex cover, and thus will be able to prove the correctness of the algorithm using induction. Our algorithm gives both a hidden set $H$ and a convex cover $C$, but can be implemented to only find one. The algorithm finds both of these in $O(n)$ time.

\begin{theorem}\label{thm:funnel}
For any funnel polygon $P$, a hidden set $H$ in $P$ and a convex cover $C$ of $P$ such that $|H| = |C|$ can be found in linear-time. 
\end{theorem}

\begin{proof}
The algorithm essentially determines a convex decomposition (with Steiner points) such that for each piece in the decomposition, a hidden point can be placed along either $R_1$ or $R_2$ corresponding to it. 

The algorithm considers two cases. The first, easier case, is where all of the edges $e_i, e_{n-i}$ for $i \in [0,n-t]$ strongly see each other, i.e. they form a convex quadrilateral $v_i,v_{i+1},v_{n-i},v_{n-i+1}$ that is within the polygon $P$. We will refer to this as the Case 1 and the case in which this is not true as Case 2. We give two examples of Case 1 in Figure \ref{fig:funnel_easy}. If this is the case, then we can simply place a hidden point on each midpoint of $R_1$ (the longer chain) and connect the corresponding edges $e_i, e_{n-i}$ until there are no more edges to connect across with from $R_2$. At this point, we can simply partition the remaining region into triangles between the remaining edges of $R_1$ and the vertex $v_{t+1}$.

\begin{figure}[ht]
    \centering
    \includegraphics[width=.5\textwidth]{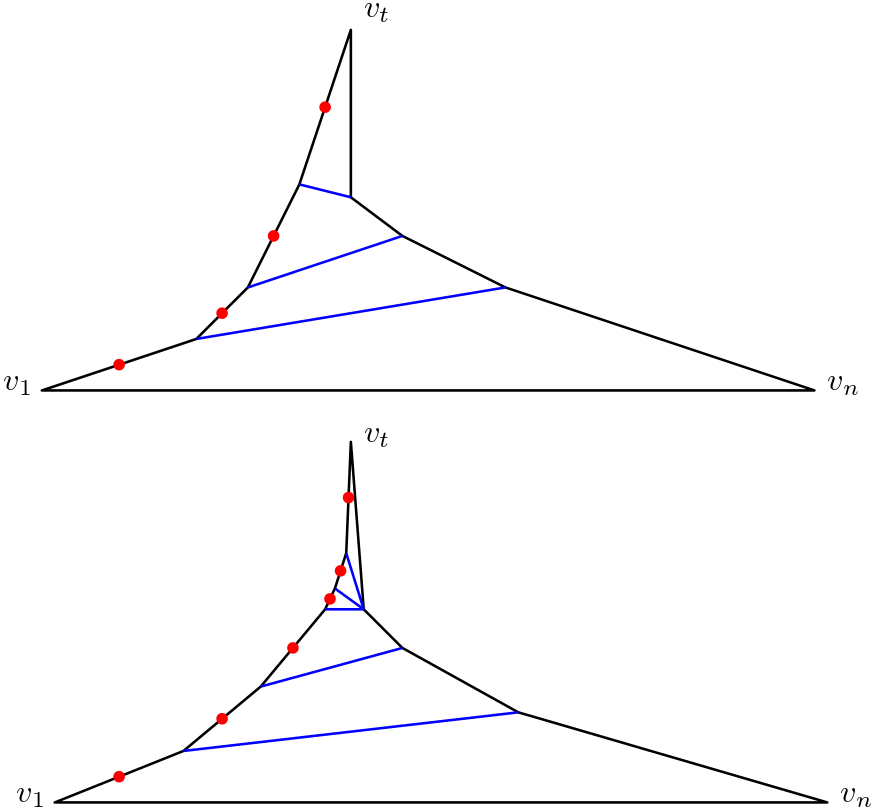}
    \caption{Case 1 for the funnel polygons.}
    \label{fig:funnel_easy}
\end{figure}

From here, we can use an induction on $n$, the number of vertices, to cover Case 2 where not all edge pairs satisfy the strong visibility requirement of Case 1. Clearly, when $n=3$, we have an instance of Case 1 and thus already have a solution. From here we just need to show that if it holds for $n \in [3,k]$, then it also holds for $n = k+1$. We can assume that for when $n=k+1$, we do not have an instance of Case 1 since we have already shown that. Let $e_i,e_{n-i}$ be the lowest $i$ pair for which $v_i,v_{i+1},v_{n-i},v_{n-i+1}$ is not within $P$. 

If $v_i$ does not see $v_{n-i+1}$, then the pair $e_{i-1} e_{n-i+1}$ would also not have their quadrilateral within $P$ and $i-1 < i$, so $v_i$ sees $v_{n-i+1}$. If $v_{i+1}$ does not see $v_{n-i}$, then either $v_i$ does not see $v_{n-i}$ or $v_{n-i+1}$ does not see $v_{i+1}$ because if there is an internal vertex $u$ in the shortest path from $v_{i+1}$ to $v_{n-i}$, it must be from either the chain of $R_1$ after $v_{i+1}$ or the chain of $R_2$ before $v_{n-i}$. This follows from \cite[Lemma 1]{GhoshWeakVisible}, since funnel polygons are a subclass of polygons weakly visible from a convex edge. Because $R_1$ is a reflex chain, if $u \in R_1$ then $u$ must also be in the shortest path from $v_i$ to $v_{n-i}$. If $u \in R_2$, then $u$ must be in the shortest path from $v_{n-i+1}$ to $v_{i+1}$. See Figure \ref{fig:funnel_block} for an example. Since these shortest paths have internal vertices, the endpoints cannot see each other.

\begin{figure}[ht]
    \centering
    \includegraphics[width = .5\textwidth]{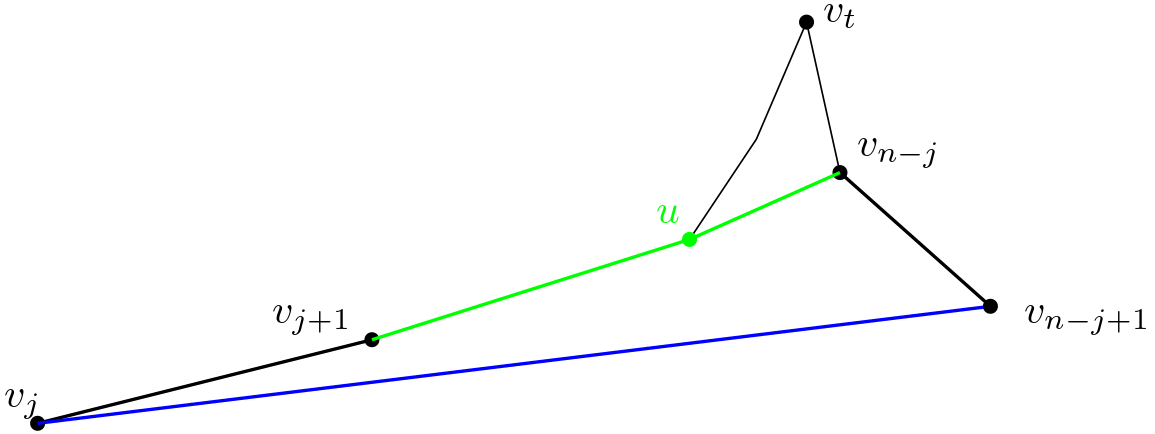}
    \caption{Lemma 1 of \cite{GhoshWeakVisible} implies that if $v_{i+1}$ does not see $v_{n-i}$ then either $v_i$ does not see $v_{n-i}$ or $v_{n-i+1}$ does not see $v_{i+1}$.}
    \label{fig:funnel_block}
\end{figure}

\begin{figure}[ht]
    \centering
    \includegraphics[width = .5\textwidth]{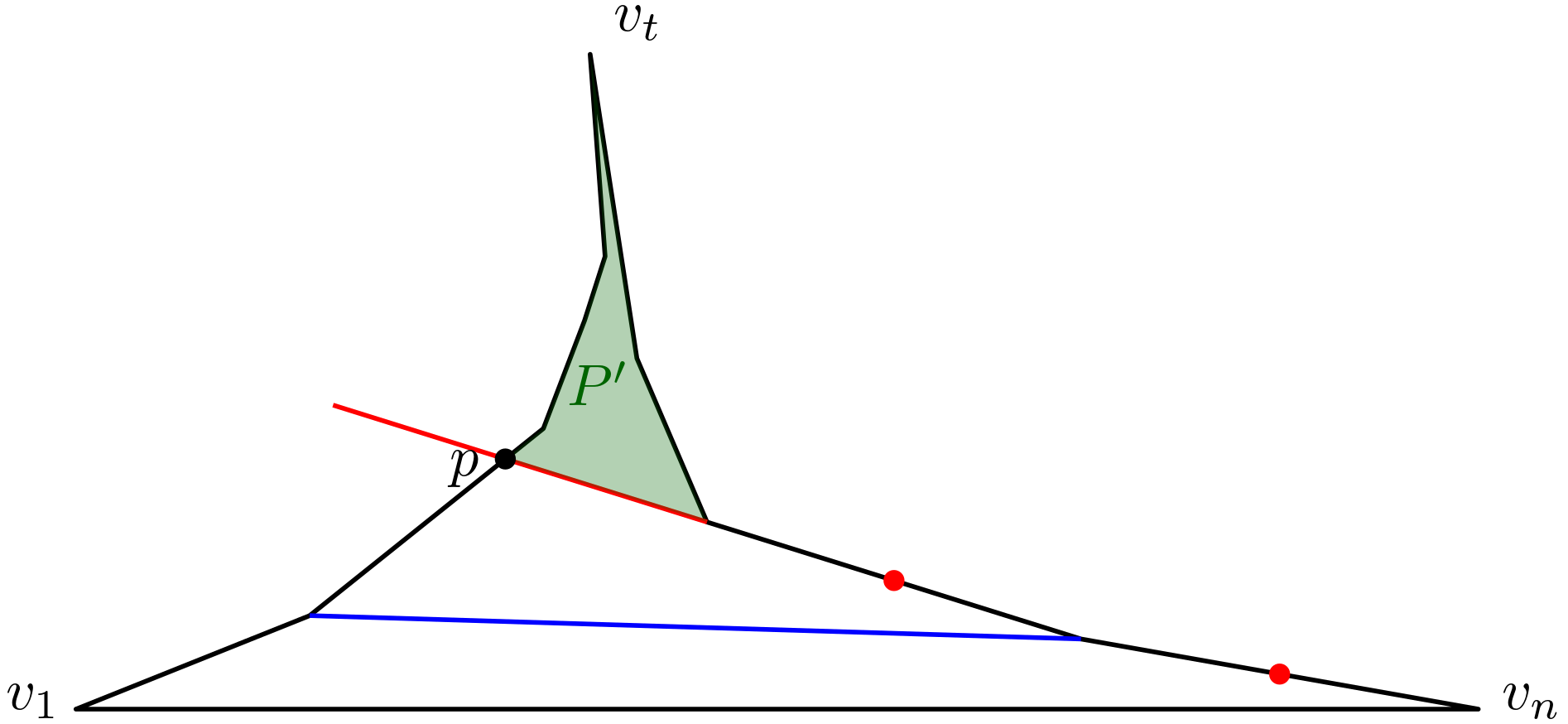}
    \caption{Case 2 for funnel polygons.}
    \label{fig:funnel_recurse}
\end{figure}

Therefore, since in all cases we find a convex cover and hidden set of the same size, Theorem \ref{thm:funnel} holds. The algorithm runs in $O(n)$ time. First, note that determining whether we are in Case 1 or Case 2 only takes $O(1)$ time per edge in $P$. This is because the only way to be in Case 2 is if the edge extension of one edge intersects with another edge, which takes $O(1)$ to determine. We do not need to compare with any other edges of the polygon. $P'$ has complexity at most 1 less than $P$, which is only the case if the index $i$ at which we form $P'$ is 1, giving us the following recurrence:
$$T(n) = T(n-1) + O(1)$$

The base case is $T(1) = O(1)$, which solves to a total runtime of $O(n)$.

\end{proof}

\subsection{Pseudotriangles}\label{sec:pseudtri}
Pseudotriangles are defined as simple polygons containing exactly 3 convex vertices, which means that it can be represented as three reflex chains $R_1 = v_1,v_2,...,v_t$, $R_2 = v_t, ... v_{s-1}, v_s$, $R_3 = v_s, . . .  v_{n-1}, v_n$. The three convex vertices here are $v_1, v_t,$ and $v_s$. Without loss of generality, we will consider the cardinalities of the chains to be such that $R_1 \geq R_2 \geq R_3$. Funnel polygons are a subclass of pseudotriangles, where $|R_3| = 1$. We will only consider pseudotriangles that are not funnel polygons, as Section \ref{sec:funnel} already gives a linear-time algorithm for funnel polygons.

\begin{theorem}\label{thm:pseudo}
For any pseudotriangle $P$, a hidden set $H$ in $P$ and a convex cover $C$ of $P$ such that $|C| \leq 2|H|$, i.e. a 2-approximation, can be found in linear-time .
\end{theorem}

\begin{proof}

We show that by using the algorithm for funnel polygons in Section \ref{sec:funnel}, we can achieve a linear-time 2-approximation for pseudotriangles. The procedure is simple: Consider the edge $e_1 = v_1v_2$, and its extension. The extension of $e_1$ will land at some point $p$ on $R_2$ called $p$. It cannot land on either $R_1$ or $R_3$ otherwise the landed on chain would not be reflex, and it must intersect some part of the boundary of $P$. Let $p$ rest on edge $e_i = v_i v_{i+1}$. This partitions $P$ into two subpolygons $P_1 = v_2,v_3 ...v_i, p$ and $P_2 = p, v_{j+1}, ... v_1$. Both $P_1$ and $P_2$ are composed of two reflex chains and a convex edge, making them both funnel polygons. This is depicted in Figure \ref{fig:pseudo_example}.

\begin{figure}[ht]
    \centering
    \includegraphics[width=.3\textwidth]{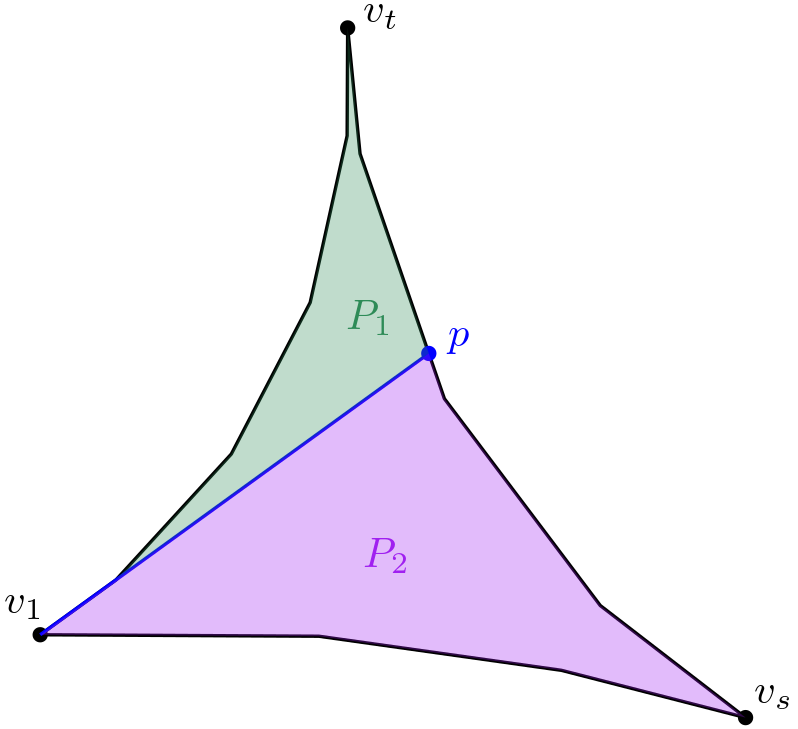}
    \caption{Pseudotriangle partition into funnel polygons that allow for a 2-approximation.}
    \label{fig:pseudo_example}
\end{figure}

From Theorem \ref{thm:funnel}, we know that we can obtain from $P_1$ a hidden set $H_1$ and a convex cover $C_1$. By symmetry, $P_2$ we get $H_2$ and $C_2$. As a hidden set, we take the larger, $H = \text{max}(H_1,H_2)$. As a convex cover, we take the union, $H = C_1 \cup C_2$. Because $|H_1| = |C_1|$ and $|H_2| = |C_2|$, we know that $|C| \leq 2 |H|$. This proves our claim, and since finding $p$ takes at most $O(n)$ time, as does each funnel, the total runtime is $O(n)$.
\end{proof}

Note, that while this algorithm does not achieve an exact answer, it is impossible to have an algorithm that only places hidden points on the edges of a pseudotriangle and that finds the exact maximum hidden set. The famous GFP (Godfried's Favorite Polygon) counterexample suffices here, and we depict it in Figure \ref{fig:godfried_example}. We give a convex cover of the boundary of size 3, which means that any hidden set constrained to the boundary must have size at most 3. We also give a hidden set of size 4, which means that 3 pieces cannot suffice for a convex cover.

\begin{figure}[ht]
    \centering
    \includegraphics[width=.3\textwidth]{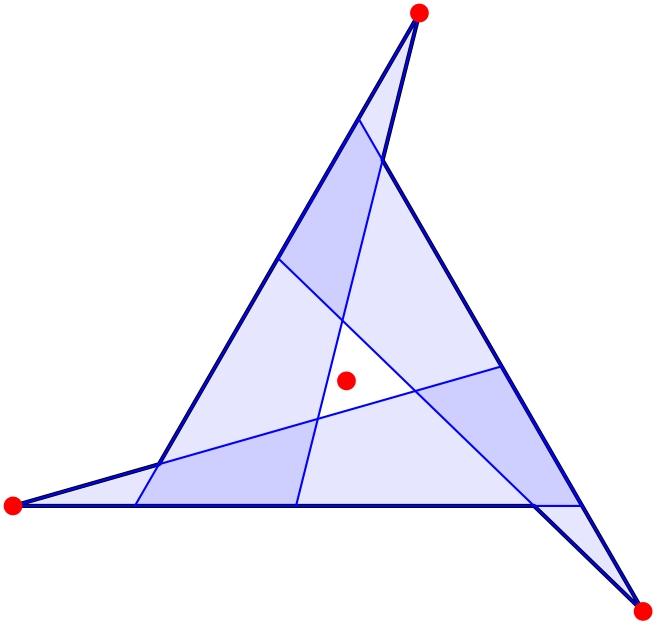}
    \caption{Godfried's favorite polygon, with a convex cover of the boundary of size 3 and a hidden set of size 4.}
    \label{fig:godfried_example}
\end{figure}

Therefore, any algorithm that places hidden points only on the boundary of a pseudotriangle cannot achieve an approximation factor better than $4/3$.

\section{Conclusion}

\begin{table}[ht]
    \centering
    \begin{tabular}{|c|c|c|} \hline
        Polygon class &  Inequality & Worst known ratio $(cc(P)/hs(P))$\\ \hline
  Histogram polygons \cite{browne2022collapsing} & $hs(P) \leq cc(P) \leq hs(P)$ & 1\\ \hline
   Spiral polygons\cite{browne2022collapsing} & $hs(P) \leq cc(P) \leq hs(P)$ & 1\\ \hline
   Funnel\cite{browne2023convex} & $hs(P) \leq cc(P) \leq hs(P)$ & 1 \\ \hline
        Monotone Mountains (or Rocking Chairs) & $hs(P) \leq cc(P) \leq  hs(P)$ &1 \\ \hline
        Convex Fans & $hs(P) \leq cc(P) \leq  hs(P)$ &1 \\ \hline
        Star-shaped polygons & $hs(P) \leq cc(P) \leq 2\cdot hs(P)$ & 3/2 \\ \hline
        Simple orthogonal polygons & $hs(P) \leq cc(P) \leq 2\cdot hs(P)$ & 4/3 \\ \hline
        Pseudotriangles\cite{browne2023convex} & $hs(P) \leq cc(P) \leq 2 \cdot hs(P)$ & 1 \\ \hline
         Polygons weakly visible from a convex edge & $hs(P) \leq cc(P) \leq 2\cdot hs(P)$ &1 \\ \hline
         Weak visibility polygons & $hs(P) \leq cc(P) \leq 4\cdot hs(P)$ &3/2\\ \hline
         
         Simple polygons & $hs(P) \leq cc(P) \leq 8\cdot hs(P)$ &3/2 \\ \hline
         Polygons with holes & $hs(P) \leq cc(P) \leq  h \cdot hs(P)$ &3/2  \\ \hline
         
    \end{tabular}
    \caption{The bounds between hidden set and convex cover for various polygon types ($h$ is the number of holes)}
    \label{tab:subclasstable}
\end{table}

We showed several extensions of previous positive and negative results for the minimum convex cover problem and the maximum hidden set problem. In addition to the algorithmic results, we give their implied combinatorial results in Table~\ref{tab:subclasstable}. The most pressing open problem with respect to both problems is whether or not there exists a constant factor approximation for convex cover in polygons with holes or there is additional hardness of approximation. It is clear from the $n^{epsilon}$ hardness for maximum hidden set that any attempt to use maximum hidden set as a bound as is done for the simple polygon case is doomed to fail unless $P=NP$. Additionally, due to the runtime for the $O(\log n)$ approximation being prohibitively large for implementation, it would also be interesting to see if there can be an improvement of the runtime for that.

With respect to hidden set, the most interesting problem is whether the problem is in fact in NP. Currently, it is not known whether the problem is in NP or if it is $\exists \mathbb{R}$-hard. Most proofs of $\exists \mathbb{R}$-hardness for geometric problems show, as a byproduct, that even polygons with rational coordinates can sometimes admit solutions which involve irrational points. However, there always exists solutions to the maximum hidden set problem in which all of the points have rational coordinates. This is because in order to be hidden, each point must have a corresponding neighborhood around it of nonzero radius for which none of the other hidden points can see any part of. This allows for the iterative transformation of any set of hidden points into one which has entirely rational coordinates. This discourages the chances that solutions to all $\exists \mathbb{R}$ formulas can be expressed as maximum hidden sets. However, there is no guarantee that these rational coordinates have fixed complexity with respect to the input coordinates, so it is also not entirely clear how to prove NP-membership.

\bibliographystyle{plain}
\bibliography{main}

\end{document}